\documentclass[11pt]{article}

\usepackage{amsmath}
\usepackage[mathscr]{eucal}
\usepackage{amssymb}
\usepackage{graphicx}
\usepackage{latexsym}
\usepackage{amsthm}

\usepackage{tikz}

\usepackage{color}
\usepackage{framed}

\ifx\suppress\undefined
\newcommand{\TODO}[1]{
\typeout{WARNING!!! there is still a TODO left}
\marginpar{\textbf{!TODO: }\emph{#1}}
}
\else
\newcommand{\TODO}[1]{}
\fi

\ifx\suppress\undefined
\newenvironment{todo}[1]{\noindent\rule{.3\textwidth}{1pt}\TODO{#1}\\}{\\\rule{.3\textwidth}{1pt}}
\else

\fi

\ifx\suppress\undefined
\newcommand{\NOTE}[1]{
\typeout{WARNING!!! there are still DRAFT NOTES left}
\marginpar{!DRAFT}\emph{\textbf{DRAFT NOTES:} #1}
}
\else
\newcommand{\NOTE}[1]{}
\fi

\newtheorem{theorem}{Theorem}
\newtheorem{lemma}{Lemma}
\newtheorem{definition}{Definition}
\newtheorem{corollary}{Corollary}

\begin{document}

\title{Optimal accessing and non-accessing structures\\ for graph protocols}
\author{Sylvain Gravier\textsuperscript{\small 1,2,}\footnote{Sylvain.Gravier@ujf-grenoble.fr}~, J\'er\^ome Javelle\textsuperscript{\small 3,}\footnote{Jerome.Javelle@imag.fr}~, Mehdi Mhalla\textsuperscript{\small 1,3,}\footnote{Mehdi.Mhalla@imag.fr}~, Simon Perdrix\textsuperscript{\small 1,3,}\footnote{Simon.Perdrix@imag.fr}}
\date{\small \textsuperscript{1} CNRS\\  \textsuperscript{2} Institut Fourier, University of Grenoble, France\\ \textsuperscript{3} LIG, University of Grenoble, France  }

\maketitle

\\
\begin{abstract}
An accessing set in a  graph is a subset $B$ of vertices such that $\exists D\subseteq B$, $\forall v\in V\setminus B, |\mathcal{N}(v)\cap D|=0 \mod 2$. 
In this paper, we introduce new bounds on the minimal size $\kappa'(G)$ of an accessing set,  and on the maximal size $\kappa(G)$ of a non-accessing set of a graph $G$. We show strong connections with perfect codes and give explicitly $\kappa(G)$ and $\kappa'(G)$ for several families of graphs. Finally, we show that the corresponding decision problems are NP-Complete. 

\end{abstract}

%%%%%%%%%%%
\section{Introduction}
%%%%%%%%%%%

In the field of quantum information theory, a very powerful tool has emerged: graph states, that are quantum entangled states which can be represented by undirected graphs \cite{HEB}.
Graph states provide a universal model of computation \cite{BRQC} \cite{BRcomp}, and  are also useful  to build several kinds of protocols \cite{blind} \cite{MS}.

The graph state formalism gives rise to strong connexions with graph theory.
In several cases, solving some problems in quantum information theory and quantum cryptography can be reduced to graph problems.

For instance, the complexity of the preparation of graph states strongly depends on the minimal degree up to local complementation on graphs \cite{HMP}, which has been proven to be the size of the smallest $D \cup Odd(D)$ where $D$ is a non-empty set of vertices and its odd neighborhood $Odd(D) = \{ v \in V , |\mathcal{N}(v)\cap D|=1\bmod{2} \}$.
A measure of entanglement is the rank width  of the graph \cite{Vdn}, the depth complexity of a quantum measurement based computation is characterized by a flow in the underlying graph \cite{flow}.

The graph parameters we investigate in this paper come from Quantum Secret Sharing (QSS) protocols using graphs \cite{MS}  where a quantum secret is shared among a set  of players and only some subsets of players can recover the secret .

Given an undirected graph $G$ of order $n$, we investigate three quantities that are strongly related to these protocols: $\kappa'(G)$ which is the smallest set containing a set of odd cardinality and its odd neighborhood, $\kappa(G)$ which is the largest set that is the odd neighborhood of a disjoint set and  $\kappa_Q(G)$ which is the maximum of $\kappa(G)$ and  $n-\kappa'(G)$ (see section \ref{Definitions} for a formal definition).

These quantities can be used to build threshold $(n,k)$ (classical or quantum) secret sharing protocols (the notion of threshold secret sharing scheme appeared in \cite{Shamir} and \cite{Blakley}), which  are protocols where among $n$ players, any set of at least $k$ players  can recover the secret, whereas any set of less than $k$ players cannot get any information about the secret. They can also be used to build "ramp secret sharing schemes" \cite{ramp} where any sufficiently large set can recover the secret whereas all small enough sets have no information about the secret (other sets of intermediate size, however, might be able to get some partial information about the secret).

For instance, in \cite{QSS_GS}, using graph based protocols where each player receives a random key and a secret encoded with the keys of his neighbors \cite{GSS},
 it has been proven, that given a graph $G$ with $n$ vertices, it is possible to define a threshold QSS scheme $(n,\kappa_Q(G)+1)$. The main idea is that $\kappa(G)$ is related to the largest set that cannot recover a classical secret, $\kappa'(G)$ to the smallest set that can recover a classical secret and $\kappa_Q(G)$ to the largest set that cannot recover a quantum secret.

First, we investigate the evolution of $\kappa$ and $\kappa'$ when we take multiple copies of a graph.
Then  we present a family of graphs  from which it is possible to build a QSS protocol with a good threshold: $(n, n-\sqrt{n})$.
This threshold is given by the value of $\kappa_Q$.
Finally, we  provide some general bounds on these quantities and prove the NP-completeness of the corresponding decision problems using reductions to the perfect code problem. 

%%%%%%%%%%
\section{Definitions}
%%%%%%%%%%
\label{Definitions}

\begin{definition}
\label{accessing}
$B\subseteq V(G)$ is \emph{accessing} if $\exists D\subseteq B$ such that $ Odd(D)\subseteq B$ and $|D| = 1\bmod{2}$, where $Odd(D)=\{v\in V(G)~\big |~|\mathcal N(v)\cap D|=1 \mod 2\}$ is the odd-neighborhood of $D$.  
\end{definition}

\begin{lemma}[\cite{QSS_GS}]
\label{non-accessing}
$B\subseteq V(G)$ is not accessing $\Leftrightarrow$ $\exists C\subseteq V\setminus B$ s.t. $B \subseteq Odd(C)$
\end{lemma}

\begin{definition}
For a given graph $G$, let
\begin{align*}
\kappa(G) &= \max_{B~\text{not accessing}} |B| \\
\kappa'(G) &=  \min_{B~\text{accessing}} |B| \\
\kappa_Q(G) &= \max \big( \kappa(G), n-\kappa'(G) \big)
\end{align*}
\end{definition}

Definition \ref{accessing} and Lemma \ref{non-accessing} are linked to QSS schemes using graph states \cite{QSS_GS}.
Indeed, accessing sets on graphs represent sets of players who can recover the secret, whereas non-accessing sets correspond to sets of players who have no information about the secret.
From any graph, we can build a QSS protocol for which any set of players of size $> \kappa_Q(G)$ can recover a quantum secret, and any set of players of size $\leq \kappa_Q(G)$ can not get any information about the secret quantum state.

%%%%%%%%%%%%%%%%%%%%%%%%%%%%%%%%
\section{ $\kappa$ and $\kappa'$ for some graph constructions}
%%%%%%%%%%%%%%%%%%%%%%%%%%%%%%%%

First we investigate the evolution of $\kappa$ and $\kappa'$ for the family of graphs $G^r$ consisting in $r$ disconnected copies of  a graph $G$. This construction is used in the NP-completeness proof of the decision problem associated with $\kappa_Q$ (see theorem \ref{NPkappaQ}).

\begin{lemma}For any graph $G$ and any $r>0$,
\begin{align}
\kappa(G^r) &= r.\kappa(G) \label{kappa_r} \\
\kappa'(G^r) &= \kappa'(G) \label{kappa_p_r}
\end{align}
where $G^1=G$ and $G^{r+1}=G\cup G^r$.
\end{lemma}

\begin{proof}~

\begin{itemize}
\item $[\kappa(G^r) = r.\kappa(G) \label{kappa_r}]$:
Let $B$ be a non-accessing set in $G$ of size $\kappa(G)$. $B$ is in the odd neighborhood of some $C \subseteq V(G)$. Then the set $B_r \subseteq V(G^r)$ which is the union of sets $B$ in each copy of the graph $G$ is in the odd neighborhood of $C_r \subset V(G^r)$, the union of sets $C$ of each copy of $G$.
Therefore $B_r$ is non-accessing and $\kappa(G^r) \geq r.\kappa(G)$.
Now if we pick any set $B_0 \subseteq V(G^r)$ verifying $|B_0|> r.\kappa(G)$, there exists a copy of $G$ such that $|B_0 \cap G| > \kappa(G)$.
Therefore $B_0$ is an accessing set and $\kappa(G^r) \leq r.\kappa(G)$.

\item $[\kappa'(G^r) = \kappa'(G) \label{kappa_p_r}]$:
Let $B$ be an accessing set in $G$ of size $\kappa'(G)$. If we consider $B$ as a subset of $V(G^r)$ contained in one copy of the graph $G$, $B$ is an accessing set in $G^r$.
Therefore $\kappa'(G^r) \leq \kappa'(G)$.
If we pick any set $B \subseteq V(G^r)$ verifying $|B|< \kappa'(G)$, its intersection with each copy of $G$ verifies $|B \cap G| < \kappa'(G)$.
Thus, each such intersection is in the odd neighborhood of some $C_i$. So $B$ is in the odd neighborhood of $\bigcup_{i=1..r}C_i$. 
Consequently, $B_0$ is a non-accessing set in $G^r$ and $\kappa'(G^r) \geq \kappa'(G)$.
\end{itemize}
\end{proof}

Now we exhibit a family of graphs on $n$  vertices with $\kappa_Q= n-\sqrt{n}$. Let $G_{p,q}$ be the complete $q$-partite graph where each independent set is of size $p$ ($G_{p,q}$ is equivalently the complement of $q$ copies of $K_p$). $G$ has $n=pq$ vertices.
\begin{lemma}
If $q=1\bmod{2}$ then $\kappa(G_{p,q})=n-p$ and $\kappa'(G_{p,q})=q$.
\end{lemma}

\begin{proof}~

\begin{itemize}
\item $[\kappa(G)\ge n-p]$: The subset $B$ composed of all the vertices but a maximal independent set (MIS) -- i.e. an independent set of size $p$ -- is in the odd neighborhood of each vertex in $V\setminus B$. So, according to Lemma \ref{non-accessing}, $B$ of size $(q-1)p=n-p$ is non accessing, as a consequence $\kappa(G)\ge n-p+1$.

\item $[\kappa(G)\le n-p]$: Any set $B$ s.t. $|B|>n-p$ contains at least one vertex from each of the $q$ MIS, i.e. a clique of size $q$. Let $D\subseteq B$ be such a clique of size $|D|=q = 1\bmod{2}$. Notice that  $Odd(D)=\emptyset$ since every vertex $v$ of the graph is connected to all the elements of $D$ but the one in the same MIS as $v$. As a consequence, $B$ is accessing.
 
\item  $[\kappa'(G)\le q]$: $B$ composed of one vertex from each MIS is an accessing set (see previous item).
\item  $[\kappa'(G)\ge q]$: If $|B|<q$ then $B$ does not intersect all the MIS of size $p$, so $B$ is in the odd neighborhood of each vertex of such a MIS. So according to Lemma \ref{non-accessing}, $B$ is not accessing. 
\end{itemize}
\end{proof}

\begin{lemma}
If $q=0\bmod{2}$ then $\kappa(G_{p,q})=max(n-p,n-q)$ and $\kappa'(G_{p,q})=p+q+1$
\end{lemma}

\begin{proof}
~
\begin{itemize}
\item $[\kappa(G)\ge max(n-p,n-q)]$: For $\kappa(G)\ge n-p$, see previous lemma. The subset $B$ composed of all the vertices but a clique of size $q$ (one vertex from each MIS) is in the odd neighborhood of $V\setminus B$. Indeed each vertex of $B$ is connected to $q-1=1\bmod{2}$ vertices of $V\setminus B$. So, according to Lemma  \ref{non-accessing}, $B$ of size $n-q$ is not accessing, as a consequence $\kappa(G)\ge n-q$.

\item $[\kappa(G)\le max(n-p,n-q)]$: Any set $B$ s.t. $|B|> max(n-p,n-q)$ contains at least one vertex from each MIS and moreover it contains a MIS $S$ of size $q$. Let $D\subseteq B\setminus S$ be a clique of size $q-1=1\bmod{2}$. Every vertex $u$ in $V\setminus B$ is connected to all the vertices in $D$ but  one, so $Odd(D)\subseteq B$.

\item  $[\kappa'(G)\le p+q-1]$: Let $S$ be an MIS. Let $B$ be the union of $S$ and  of a clique of size $q$. Let $D= B\setminus S$. $|D|=q-1=1\bmod{2}$. Every vertex  $u$ in $V\setminus B$ is connected to all the vertices of $D$ but one, so $Odd(D)\subseteq B$.

\item  $[\kappa'(G)\ge p+q-1]$: Let $|B|<p+q-1$. If  $B$ does not intersect all the MIS of size $p$, then $B$ is in the odd neighborhood  of each vertex of such a non intersecting MIS. If $B$ intersects all the MIS then it does not contain any MIS, thus there exists a clique $C\subseteq V\setminus B$ of size $q$. Every vertex in $B$ is in the odd neighborhood of $C$. 
\end{itemize}
\end{proof}

\begin{corollary}
If $n=p^2$, $\kappa_Q(G_{\sqrt{n},\sqrt{n}})= n-\sqrt{n}$
\end{corollary}

%%%%%%%%%%%%%%%%%%%%
\section{Bounds and NP-Completeness}
%%%%%%%%%%%%%%%%%%%%

For a given graph $G$, we show that the sum of $\kappa(G)$ and $\kappa'(\overline G)$ is always greater than the order of the graph $G$. The proof is based on the duality property  that the complement of an accessing set in $G$ is a non accessing set in $\overline G$:

\begin{lemma}\label{lem:dual}
Given a graph $G$, if $B$ is accessing in $G$ then $V\setminus B$ is not accessing in $\overline G$.  More precisely, if $\exists D\subseteq B$, $|D|=1\bmod 2$ and $Odd_G(D)\subseteq B$ then $Odd_{\overline G} (D)\supseteq V\setminus B$. 
\end{lemma}

\begin{proof} Let $B$ be an accessing set in $G$. $\exists D\subseteq B$ s.t. $|D|=1\bmod 2$ and $Odd_G(D)\subseteq B$. As a consequence, $\forall v\in V\setminus B$, $|\mathcal N_G(v)\cap D|=0 \bmod 2$. Since $|D|=1\bmod 2$, $\forall v\in V\setminus B$, $|\mathcal N_{\overline G}(v)\cap D|=1\bmod 2$. Thus, according to Lemma \ref{non-accessing}, $V\setminus B$ is not accessing in $\overline G$. 
\end{proof}

\begin{theorem}\label{thm:dual}
For any graph $G$ of order $n$, $$\kappa'(G)+\kappa(\overline G)\ge n$$
\end{theorem}

\begin{proof}
It exists an accessing set $B\subseteq V(G)$ s.t. $|B|= \kappa'(G)$. According to Lemma \ref{lem:dual}, $V\setminus B$ is not accessing in $\overline G$, so $n-|B|\le \kappa(\overline G)$, so $n-\kappa'(G)\le \kappa(\overline G)$. 
\end{proof}

For a given graph $G$, the closed neighborhood ($\{v\} \cup \mathcal{N}(v)$) of any vertex $v$ is an accessing set, whereas the open neighborhood ($\mathcal{N}(v)$) of any vertex $v$ is a non accessing set, as a consequence:
$$\kappa(G)\ge \Delta~~~~~~~~~~~~~~~~~\kappa'(G)\le \delta+1$$
where $\Delta$ (resp. $\delta$) denotes the maximal (resp. minimal) degree of the graph $G$. 

In the following, we prove an upper bound on $\kappa(G)$ and a lower bound on $\kappa'(G)$.

\begin{lemma}\label{lem:ubound}
For any graph $G$, $$\kappa(G)\le \frac{n.\Delta}{\Delta +1}$$ where $n=|V(G)|$.
\end{lemma}

\begin{proof}
Let $B\subseteq V(G)$ be a non accessing set, so according to Lemma \ref{non-accessing}, $\exists C\subseteq V\setminus B$ s.t. $Odd(C)\supseteq B$. $|C|\le n-|B|$ and $|B|\le |Odd(C)|\le \Delta.|C|$, so $|B|\le \Delta.(n-|B|)$. It comes that $|B|\le \frac{n.\Delta}{\Delta+1}$, so $\kappa(G)\le \frac{n.\Delta}{\Delta+1}$.
\end{proof}

This bound is reached only for graphs having a perfect code. A graph $G$ has a perfect code if  it exists $ C\subseteq V(G)$ such that $C$ is an independent set and every vertex in $V(G)\setminus C$ has exactly one neighbor in $C$.  

\begin{theorem}
For any graph $G$, $\kappa(G)=\frac{n.\Delta}{\Delta +1}$ iff $G$ has a perfect code $C$ such that  $\forall v\in C$, $d(v)=\Delta$. 
\end{theorem}

\begin{proof}
($\Leftarrow$) Let $C$ be a perfect code of $G$ s.t.  $\forall v\in C$, $\delta(v)=\Delta$. $V(G)\setminus C$ is a non accessing set since $Odd(C) = V(G)\setminus C$. Moreover $|V(G)\setminus C|= \frac{n\Delta}{\Delta+1}$, so $\kappa(G)\ge \frac{n.\Delta}{\Delta +1}$. According to Lemma \ref{lem:ubound}, $\kappa(G)\le \frac{n\Delta}{\Delta+1}$, so $\kappa(G)= \frac{n\Delta}{\Delta+1}$.  \\
($\Rightarrow$) Let $B$ be a non accessing set of size $\frac{n.\Delta}{\Delta +1}$. According to Lemma \ref{non-accessing}, $\exists C\subseteq V\setminus B$ s.t. $Odd(C)\supseteq B$. Notice that $|C|\le n-\frac{n.\Delta}{\Delta +1}=\frac{n}{\Delta +1}$. Moreover $|C|.\Delta\ge |Odd(C)| \ge  |B|$, so  $|C|=\frac{n}{\Delta +1}$. It comes  $|B|=|B\cap Odd(C)|\le \sum_{v\in C}d(v)\le \Delta.\frac{n}{\Delta+1}= |B|$. Notice that if $C$ is not a perfect code the first inequality is strict, and if $\exists v\in C$, $d(v)<\Delta$, the second inequality is strict. Consequently, $C$ is a perfect code and $\forall v\in C$, $d(v)=\Delta$.  
\end{proof}

\begin{corollary}\label{cor:kappa}
Given a $\Delta$-regular graph $G$, 
$$\kappa(G) = \frac {n\Delta} {\Delta+1}\iff G \text{ has a perfect code} $$
\end{corollary}

We consider the problem {\bf KAPPA$_\le$} (resp. {\bf KAPPA$_\ge$}) which consists in deciding, given a graph $G$ and an integer $k\ge 0$, whether $\kappa(G)\le k$ (resp. $\kappa(G) \ge k$).

\begin{theorem}
{\bf KAPPA$_\ge$} is NP-Complete.
\end{theorem}

\begin{proof}
${\bf KAPPA_\ge}$ is in the class NP since a non accessing set of size $k$ is a YES certificate.
For the completeness, given a 3-regular graph, if $\kappa(G) \ge \frac 3 4 n$ then $\kappa(G) =  \frac 3 4 n$ (since $\kappa(G) \le  \frac{n\Delta}{\Delta+1}$ for any graph). Moreover, according to Corollary \ref{cor:kappa}, $\kappa(G)= \frac 3 4 n$ iff $G$ has a perfect code. Since the problem of deciding whether a $3$-regular graph has a perfect code  is known to be NP complete (see \cite{perf_codes} and \cite{KMP}), so is ${\bf KAPPA_\ge}$.
\end{proof}

\begin{corollary}
{\bf KAPPA$_\le$} is coNP-Complete.
\end{corollary}

Now we introduce a lower bound on $\kappa'$.

\begin{lemma}
For any graph $G$, 
$$\kappa'(G)\ge \frac n{n-\delta}$$ where $\delta$ is the minimal degree of $G$. 
\end{lemma}

\begin{proof} According to Theorem \ref{thm:dual}, $\kappa'(G)\ge n-\kappa(\overline G)$. Moreover, thanks to Lemma \ref{lem:ubound}, $n-\kappa(\overline G) \ge  n-\frac{n\Delta({\overline G})}{\Delta({\overline G})+1} =n- \frac{n(n-1-\delta(G))}{n-\delta(G)}= \frac{n}{n-\delta}$. 
\end{proof}

This bound is reached for regular graphs so that their complementary graph has a perfect code, more precisely:

\begin{theorem}\label{thm:kappap}
Given $G$ a $\delta$-regular graph s.t. $\frac {n}{n-\delta}$ is odd: $$\kappa'({G}) = \frac n {n-\delta} \iff \overline{G} \text{ has a perfect code} $$
\end{theorem}

\begin{proof}
($\Leftarrow$) Let $C$ be a perfect code of $\overline G$. Since $|C|=\frac n{\Delta(\overline G)+1}=\frac n{n-\delta}=1~\bmod{2}$, $Odd_{{G}}(C)\subseteq C$, thus $C$ is an accessing set in ${G}$, so $ \kappa'({G})\le \frac n {n-\delta}$. Since $\kappa'(G) \ge \frac n {n-\delta}$ for any graph, $\kappa'(G) = \frac n {n-\delta}$\\
($\Rightarrow$) Let $B$ be an accessing set of size $\frac n{n-\delta}$ in ${G}$. $\exists D\subseteq B$ s.t. $|D|=1\bmod 2$ and $Odd_{{G}}(D)\subseteq B$. According to Lemma \ref{lem:dual},   $V\setminus B\subseteq Odd_{\overline G}(D)$, so $|Odd_{\overline G}(D)|\ge \Delta(\overline G)\frac{n }{n-\delta}$, which implies that  $|D|.\Delta(\overline G) \ge \Delta(\overline G)\frac{n}{n-\delta}$. As a consequence, $|D|= \frac n {n-\delta}$ and since every vertex of $V\setminus B$ (of size $\Delta(\overline G)\frac{n}{n-\delta}$) in $\overline G$ is connected to $D$, $D$ must be a perfect code. 
\end{proof}

We consider the problem {\bf KAPPA$'_\le$} (resp. {\bf KAPPA$'_\ge$}) which consists in deciding, given a graph $G$ and an integer $k\ge 0$, whether $\kappa'(G)\le k$ (resp. $\kappa'(G) \ge k$)?

\begin{theorem}
 ${\bf KAPPA'_\le}$ is NP-Complete.
\end{theorem}

\begin{proof}
${\bf KAPPA'_\le}$ is in the class NP since an accessing set of size $k$ is a YES certificate.
For the completeness, given a 3-regular graph $G$,  if $\frac{n}{4}$ is odd then according to Theorem \ref{thm:kappap}, $G$ has a perfect code iff $\kappa'(\overline G)=\frac{n}{4}$. If   $\frac{n}{4}$ is even, we add a $K_4$ gadget to the graph $G$. Indeed, $G\cup K_4$ is a 3-regular graph and $\frac{n+4}{4}=\frac{n}{4}+1$ is odd. Moreover, $G$ has a perfect code iff $G\cup K_4$ has a perfect code iff $\kappa'(\overline {G\cup K_4})=\frac{n}{4}+1$. Since deciding whether a $3$-regular graph has a perfect code  is known to be NP complete, so is ${\bf KAPPA'_\le}$
\end{proof}

\begin{corollary}
{\bf KAPPA$'_\ge$} is co-NP-Complete.
\end{corollary}

We consider the problem {\bf QKAPPA} which consists in deciding, for a given graph $G$ and $k\ge 0$, whether $\kappa_Q(G)\le k$, i.e. $\kappa(G)\le  k$ and $\kappa'(G)\ge  n-k$?

\begin{theorem}
\label{NPkappaQ}
{\bf QKAPPA} is coNP-Complete.
\end{theorem}

\begin{proof}
{\bf QKAPPA} is co-NP since a non accessing set of size $k-1$ or an accessing set of size $n-k+1$ is a NO certificate. For the completeness, we use a reduction to the problem {\bf KAPPA$'_\ge$}. Given a graph $G$ and any $k\ge0$,  $\kappa_G(G^k)\le kn-k$ iff $(\kappa(G^{k})\le kn-k \wedge \kappa'(G^k)\ge  k) \iff (\kappa(G^{k})\le kn-k \wedge \kappa'(G^k)\ge k) \iff  (\kappa(G)\le  n-1 \wedge \kappa'(G)\ge k )$. Since for any graph $G$, $V(G)$ is an accessing set, $\kappa(G)\le n-1$, thus $k-1$ is a threshold for the protocol $G^k$ iff $\kappa'(G)\ge k$. As a consequence,  {\bf QKAPPA} is coNP-Complete. 
\end{proof}

\section{Conclusion}

In this paper, we have studied the quantities $\kappa$, $\kappa'$ and $\kappa_Q$ that can be computed on graphs.
They correspond to the extremal cardinalities accessing and non-accessing sets can reach.
These quantities present strong connexions with quantum information theory and the graph state formalism, and especially in the field of quantum secret sharing.

Thus, we have studied and computed these quantities on some specific families of graphs, and we deduced they are candidates for good threshold quantum secret sharing protocols.
Then we have proven the NP-completeness of the decision problems associated with $\kappa$, $\kappa'$ and $\kappa_Q$.

A related question is still open: is the problem of deciding whether the minimal degree up to local complementation is greater than $k$ NP-complete?
This problem seems very close to finding $\kappa'$ since it consists in finding the smallest set of vertices of the form $D \cup Odd(D)$ with $D \neq \varnothing$, without the constraint of parity $|D|=1 \bmod 2$ as for $\kappa'$.

\bibliographystyle{plain}
\bibliography{Kappa}

\begin{thebibliography}{10}

\bibitem{Blakley}
George~Robert Blakley.
\newblock Safeguarding cryptographic keys.
\newblock In {\em Proceedings of the National Computer Conference}, pages
  313--317. American Federation of Information Processing Societies, 1979.

\bibitem{ramp}
George~Robert Blakley and Catherine Meadows.
\newblock Security of ramp schemes.
\newblock In {\em Advances in Cryptology, Proceedings of CRYPTO 84}, volume
  196, pages 242--268. Springer, 1984.

\bibitem{blind}
Anne Broadbent, Joseph Fitzsimons, and Elham Kashefi.
\newblock Universal blind quantum computation.
\newblock In {\em Proceedings of FOCS}, pages 517--526, 2009.

\bibitem{HEB}
Marc Hein, Jens Eisert, and Hans~J Briegel.
\newblock Multi-party entanglement in graph states.
\newblock {\em Physical Review A}, 69, 2004.

\bibitem{HMP}
Peter H{\o}yer, Mehdi Mhalla, and Simon Perdrix.
\newblock Resources required for preparing graph states.
\newblock In {\em Proceedings of ISAAC'06}, pages 638--649, 2006.

\bibitem{GSS}
J{\'e}r{\^o}me Javelle, Mehdi Mhalla, and Simon Perdrix.
\newblock Classical versus quantum graph-based secret sharing.
\newblock {\em arXiv:1109.4731}, 09 2011.

\bibitem{QSS_GS}
J{\'e}r{\^o}me Javelle, Mehdi Mhalla, and Simon Perdrix.
\newblock New protocols and lower bound for quantum secret sharing with graph
  states.
\newblock {\em arXiv:1109.1487}, 09 2011.

\bibitem{perf_codes}
Jan Kratochvil.
\newblock Perfect codes in general graphs.
\newblock {\em 7th Hungarian colloqium on combinatorics, Eger}, 1987.

\bibitem{Vdn}
Guifre~Vidal Maarten Van~de Nest, Wolfgang~D{\"u}r and Hans~J. Briegel.
\newblock Classical simulation versus universality in measurement based quantum
  computation.
\newblock {\em Physical Review A}, 75:012337, 2007.

\bibitem{MS}
Damian Markham and Barry~C. Sanders.
\newblock Graph states for quantum secret sharing.
\newblock {\em Physical Review A}, 78:042309, 2008.

\bibitem{flow}
Mehdi Mhalla and Simon Perdrix.
\newblock Finding optimal flows efficiently.
\newblock In {\em Proceedings of 35th ICALP}, pages 857--868, 09 2007.

\bibitem{BRQC}
Robert Raussendorf and Hans Briegel.
\newblock A one-way quantum computer.
\newblock {\em Physical Review Letters}, 86(22):5188--5191, 2001.

\bibitem{BRcomp}
Robert Raussendorf and Hans Briegel.
\newblock Computational model underlying the one-way quantum computer.
\newblock {\em Quantum Information and Computation}, 6:433, 2002.

\bibitem{KMP}
Uros~Milutinovic Sandi~Klavzar and Ciril Petr.
\newblock 1-perfect codes in sierpinski graphs.
\newblock {\em Bulletin of the Australian Mathematical Society}, 66:369--384,
  2002.

\bibitem{Shamir}
Adi Shamir.
\newblock How to share a secret.
\newblock {\em Communications of the ACM}, 22(11):612--613, 1979.

\end{thebibliography}

\end{document}